\documentclass[11pt]{article}
\usepackage[margin=1in]{geometry}
\usepackage{amsmath,amssymb,amsthm}
\usepackage{mathtools}
\usepackage[colorlinks=true,linkcolor=blue,citecolor=blue,urlcolor=blue]{hyperref}
\usepackage{enumitem}
\usepackage{booktabs}
\usepackage{tikz}
\usetikzlibrary{positioning,calc,decorations.pathreplacing,arrows.meta,shapes.geometric,fit,backgrounds}
\definecolor{fncol}{RGB}{0,102,153}
\definecolor{storecol}{RGB}{150,25,60}
\definecolor{fftcol}{RGB}{0,110,72}
\usepackage{framed}
\usepackage{float}

\theoremstyle{plain}
\newtheorem{theorem}{Theorem}
\numberwithin{theorem}{section}
\newtheorem{lemma}[theorem]{Lemma}
\newtheorem{claim}[theorem]{Claim}
\newtheorem{observation}[theorem]{Observation}
\newtheorem{corollary}[theorem]{Corollary}
\theoremstyle{definition}

\newcommand{\Ot}{\widetilde{O}}
\newcommand{\eps}{\epsilon}
\newcommand{\Ct}{C_{\star}}
\newcommand{\cnt}{\operatorname{count}}

\newcommand{\smallmap}{g}
\newcommand{\extendedmap}{G}

\newcommand{\cg}[1]{\operatorname{cong}(#1)}
\newcommand{\cgj}[2]{\operatorname{cong}_{#2}{(#1)}}

\usepackage[ruled,linesnumbered,vlined]{algorithm2e}
\SetKw{Break}{break}

\title{Recovery Beats Storage: Improved Space for Preprocessed 3SUM%
\footnote{A previous version of this paper (\href{https://arxiv.org/abs/2608.22355v1}{arXiv:2608.22355v1})
achieved space $\widetilde{O}(n^{\max(2-\epsilon,\,11/6-\epsilon/3)})$,
and the current version improves it to $\tilde{O}(n^{2-\epsilon})$ 
by eliminating the explicit storage of heavy values.
}
}

\author{Amir Carmel%
    \thanks{Email: \texttt{amir6423@gmail.com}
    }
    \qquad
  Yakov Kosoburd%
  \thanks{Email: \texttt{yakov.kosoburd@weizmann.ac.il}
    }
    \qquad
  Robert Krauthgamer%
  \thanks{The Harry Weinrebe Professorial Chair of Computer Science.
    Work partially supported by the Israel Science Foundation grant \#1336/23.
    Email: \texttt{robert.krauthgamer@weizmann.ac.il}    
  }\\
  Weizmann Institute of Science
}
\date{}

\begin{document}
\maketitle

\begin{abstract}
The 3SUM problem asks, given sets $A,B,C$ of integers, 
whether there exist $a\in A$ and $b\in B$ whose sum belongs to $C$. 
In the preprocessed variant with unknown $C$, 
one preprocesses sets $A$ and $B$, each of size $n$,
and subsequently answers a query specified by subsets $A'\subseteq A$, $B'\subseteq B$ and a target set $C'$ of size $O(n)$, 
by solving the 3SUM instance $(A',B',C')$.

Kirkpatrick, Kuszmaul, Mathialagan, and Vassilevska Williams [ICALP 2026] 
gave the first algorithm with subquadratic space for this problem, 
achieving $\widetilde{O}(n^{3/2+\epsilon})$ query time using $\widetilde{O}(n^{2-2\epsilon/3})$ space, for every $\epsilon\in[0,1/2]$.
Their algorithm employs separate mechanisms for heavy and light targets,
and for each heavy target it stores explicitly the list of pairs $(a,b)$ whose sum is congruent to it modulo a random prime $p$ but differs from it.
These lists dominate the space bound.

We present a unified construction that uses a single mechanism for all queries.
Instead of storing these lists of pairs, 
our algorithm recovers on demand the residue class of the target modulo a random prime
by inverting the function $(a,b)\mapsto (a+b\bmod p)$. 
Additionally, instead of storing the heavy targets, it recomputes at query time a set that covers them.
This leads to an improved space bound of $\widetilde{O}(n^{2-\epsilon})$,
while maintaining the same query time.
Moreover, our construction is the first to achieve subquadratic space 
while supporting adaptively chosen queries. 

\end{abstract}

\section{Introduction}\label{sec:intro}

In the 3SUM problem, the input is sets $A,B,C$ of $n$ integers each, 
and the goal is to determine whether there exists $(a,b,c)\in A\times B\times C$
such that $a+b = c$. 
A classical algorithm solves 3SUM in $O(n^{2})$ time, 
and for a long time algorithms improved on this bound only by subpolynomial factors~\cite{BDP,Chan}.
The problem is central to fine-grained complexity, 
where the 3SUM Hypothesis \cite{GO} asserts that 
no algorithm can solve 3SUM in $O(n^{2-\eps})$ time for fixed $\eps>0$. 
It has been one of the central conjectures in the field,
and it implies conditional quadratic lower bounds for many problems
in computational geometry, string algorithms, and dynamic data structures; see~\cite{VW} for a survey.
Very recently, the 3SUM Hypothesis was refuted~\cite{AV26},
by a deterministic algorithm that solves 3SUM on integers of polynomial size in $O(n^{1.9992})$ time,
see also~\cite{KL26,chan20263sumsimplerwithoutmatrix}.
While the 3SUM Hypothesis was refuted, the problem remains central. The best exponent for 3SUM is still unknown, and some reductions from it may still have value.

In the preprocessed variant of 3SUM, 
the sets are given as queries, but they are drawn from larger sets of size $n$ 
that are known in advance and can be preprocessed. 
The variant refines the complexity of 3SUM by separating the contribution of the universe from that of the particular instance. It thus serves as a basis both for algorithms on inputs that share a universe and for conditional lower bounds on data structures.
Bansal and Williams~\cite{BW}, who attributed the question to Avrim Blum,
initiated the study of this problem 
and gave a data structure with query time $n^{2}/\operatorname{polylog}(n)$. 
Subsequent work obtained query times that are truly subquadratic  
(i.e., smaller than $n^2$ by a polynomial factor) for two variants. 
In the \emph{known-$C$} variant, 
all three sets $A,B,C$ are preprocessed, 
and then a query specifies subsets $A'\subseteq A$, $B'\subseteq B$, and $C'\subseteq C$, 
asking whether there exists $(a,b,c)\in A'\times B'\times C'$ such that
$a+b=c$.
In the \emph{unknown-$C$} variant, only $A$ and $B$ are preprocessed, 
and the query specifies subsets $A'\subseteq A$, $B'\subseteq B$, and 
a set $C'$ of $O(n)$ integers, and asks the same question. 
We focus on the unknown-$C$ variant;
observe that it is at least as hard as the known-$C$ one,
because a data structure for the former is clearly applicable also for the latter.
Throughout, we refer to potential elements of $C'$ as \emph{targets}.

Our main result obtains the following tradeoff between the query time and the space occupied by the data structure,
where $\Ot(f)$ hides logarithmic factors in $f$.

\begin{theorem}\label{thm:main}
For every $\eps\in[0,1/2]$ 
there is a randomized data structure for preprocessed 3SUM with unknown-$C$,
whose preprocessing runs in $\Ot(n^{2})$ time and 
occupies space $\Ot (n^{2-\eps} )$.
The preprocessing succeeds with high probability, in which case 
every query is answered correctly in worst-case time $\Ot(n^{3/2+\eps})$.
\end{theorem}

At the endpoint $\eps=0$, our query time is $\Ot(n^{3/2})$ and our space bound is $\Ot(n^{2})$,
which was previously known~\cite{KPS,KKMV26}. 
At the other endpoint $\eps=1/2$, our query time is $\Ot(n^{2})$ and our space bound is $\Ot(n^{3/2})$,
although space $O(n)$ trivially suffices, and one can use the known subquadratic algorithms with no preprocessing.

For every fixed $\eps\in(0,1/2]$,
our space bound strictly improves on the best prior bound $\Ot(n^{2-2\eps/3})$ of~\cite{KKMV26}, 
discussed below.

Besides improving the known space and time bounds, 
our construction also provides a stronger type of guarantee: 
Conditioned on successful preprocessing, 
the data structure answers correctly all possible queries simultaneously, 
akin to the ``forall" (rather than ``foreach") guarantee
in sparse recovery and sketching algorithms.
It can thus be used in scenarios where an adversary chooses queries adaptively, 
or based on the data structure (i.e., as a function of the preprocessing phase).
In contrast, prior randomized algorithms can only handle an oblivious adversary~\cite{KKMV26}, 
i.e., for every fixed query the answer is correct with high probability. 
In addition, to address the possibility of unsuccessful preprocessing, we devise a check (Section~\ref{sec:cert}) that guarantees that a wrong answer is never returned.

\paragraph{Prior work.}
Chan and Lewenstein~\cite{CL} achieved the first truly subquadratic query time
via additive combinatorics 
(subsequently extended to real inputs~\cite{DBLP:conf/soda/Fischer25}). 
Chan, Vassilevska Williams, and Xu~\cite{CVX} later improved the
query-time and space bounds. 
Kasliwal, Polak, and Sharma~\cite{KPS} obtained significant improvements 
for both variants through a much simpler scheme that uses 
the fast Fourier transform (FFT) for certain counting modulo a random prime. 
Fischer and Schmidt~\cite{FS26} gave a randomized reduction from real inputs to integer inputs, 
with $\Ot(n^{2})$ preprocessing, which transfers the bounds for integers to real inputs,
in particular $\Ot(n^{3/2})$ query time with $\Ot(n^{2})$ preprocessing time.

In the unknown-$C$ variant, however, 
all data structures with subquadratic query time continued to use quadratic space, 
essentially by storing information about the full sumset $A+B$. 
This barrier was broken by~\cite{KKMV26}, which gives, for every $\eps\in[0,1/2]$, 
a data structure with $\Ot(n^{3/2+\eps})$ query time and $\Ot(n^{2-2\eps/3})$ space. 
We summarize the known bounds and tradeoffs between space and query-time
in Figure~\ref{fig:curves}. 

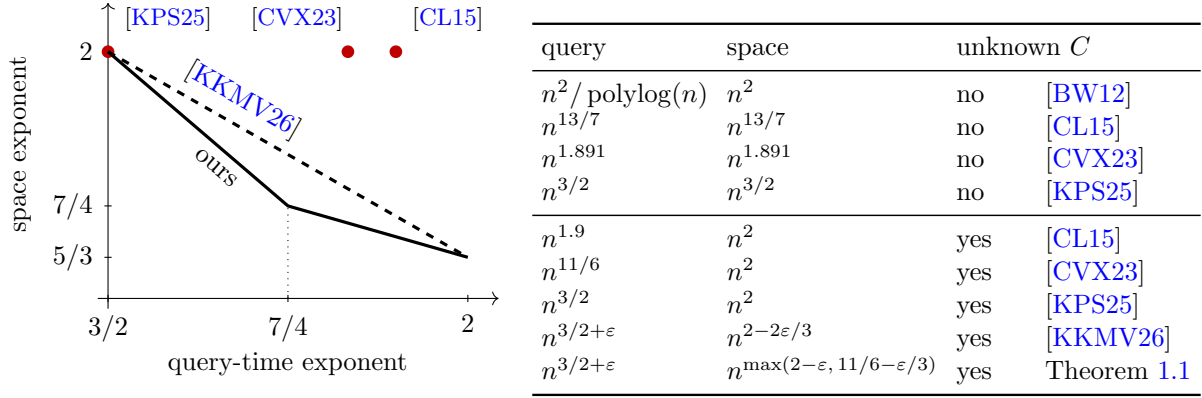
\begin{figure}[H]
\centering
\begin{minipage}[c]{0.41\textwidth}
\centering
\begin{tikzpicture}[scale=0.68]
\draw[->] (0,-0.2) -- (0,6.45);
\draw[->] (-0.2,0) -- (7.6,0);
\node[below,font=\small] at (3.5,-0.85) {query-time exponent};
\node[rotate=90,anchor=south,font=\small] at (-1.3,3.0) {space exponent};
\foreach \y/\lab in {5.5/2, 3.0/{7/4}, 2.167/{5/3}, 0.5/{3/2}}{
  \draw (0.07,\y) -- (-0.07,\y); \node[left,font=\small] at (-0.1,\y) {$\lab$};}
\foreach \x/\lab in {0/{3/2}, 3.5/{7/4},  7/2}{
  \draw (\x,0.07) -- (\x,-0.07); \node[below,font=\small] at (\x,-0.12) {$\lab$};}
\fill[red!75!black] (0,5.5) circle (3.5pt);
\node[anchor=south west,font=\footnotesize] at (0.12,5.7) {\cite{KPS}};
\fill[red!75!black] (4.667,5.5) circle (3.5pt);
\node[anchor=south east,font=\footnotesize] at (4.78,5.7) {\cite{CVX}};
\fill[red!75!black] (5.6,5.5) circle (3.5pt);
\node[anchor=south west,font=\footnotesize] at (5.72,5.7) {\cite{CL}};
\draw[very thick, dashed] (0,5.5) -- (7,2.167);
\node[font=\small,rotate=-25.5,anchor=south] at (2.4,4.36) {\cite{KKMV26}};
\draw[very thick] (0,5.5) -- (7,0.5);
\node[font=\small,rotate=-35.5,anchor=north] at (2.3,3.86) {ours};
\end{tikzpicture}
\end{minipage}\hfill
\begin{minipage}[c]{0.57\textwidth}
\centering\small
\setlength{\tabcolsep}{3.5pt}
\begin{tabular}{lll@{\qquad}l}
\toprule
query & space & \multicolumn{2}{l}{unknown $C$} \\
\midrule
$n^{2}/\operatorname{polylog}(n)$ & $n^{2}$ & no & \cite{BW} \\
$n^{13/7}$ & $n^{13/7}$ & no & \cite{CL} \\
$n^{1.891}$ & $n^{1.891}$ & no & \cite{CVX} \\
$n^{3/2}$ & $n^{3/2}$ & no & \cite{KPS} \\
\midrule
$n^{1.9}$ & $n^{2}$ & yes & \cite{CL} \\
$n^{11/6}$ & $n^{2}$ & yes & \cite{CVX} \\
$n^{3/2}$ & $n^{2}$ & yes & \cite{KPS} \\
$n^{3/2+\eps}$ & $n^{2-2\eps/3}$ & yes & \cite{KKMV26} \\
$n^{3/2+\eps}$ & $n^{2-\eps}$ & yes & Theorem~\ref{thm:main} \\
\bottomrule
\end{tabular}
\end{minipage}
\caption{Left: The tradeoff between exponents of space and of query-time, 
for the unknown-$C$ variant. 
The circles mark the prior unknown-$C$ results, all at quadratic space. 
Right: Known bounds for Preprocessed 3SUM, 
suppressing $\operatorname{polylog} (n)$ factors except for~\cite{BW}.
All rows take $\Ot(n^{2})$ preprocessing time.
}
\label{fig:curves}
\end{figure}

To reduce the space of~\cite{KPS} below quadratic, \cite{KKMV26} answer light targets
using the Fiat--Naor function inverter~\cite{FN}.
They observe that an inverter with the ideal tradeoff $S\cdot T=\Ot(N)$
between space $S$ and query time $T$, where $N$ is the domain size,
would improve their space bound to $\Ot(n^{2-\eps})$.
We attain this bound already with the Fiat--Naor inverter,
and an ideal inverter would improve our space to $\Ot(n^{2-2\eps})$
at the same query time $\Ot(n^{3/2+\eps})$,
which is the tradeoff that~\cite{KKMV26} conjectured to be optimal
(prior to the recent subquadratic 3SUM algorithm).

\subsection{Technical overview}\label{sec:tech}

We first describe two known building blocks that our data structure uses,
namely, counting modulo a random prime and function inversion. 
Let $p$ be a random prime drawn at the preprocessing phase. 
Given the query sets $A'$ and $B'$,
a single FFT computation of order $p$ produces, for all residues $r\in[p]$ at once, 
the counts 
\begin{align*}
  \cg{r} 
  &:= \bigl|\big\{ (a,b)\in A'\times B' \ :\ a+b = r \pmod p \big\}\bigr|.
\end{align*}
In comparison, the answer for each target $c$ is determined by 
the analogous non-modular count
\begin{align*}
  \cnt_{A'\times B'}(c) 
  &:= \bigl|\big\{ (a,b)\in A'\times B' \ :\ a+b=c \big\}\bigr|,
\end{align*}
namely, the answer for target $c$ is \textsc{yes} iff $\cnt_{A'\times B'}(c)>0$.
The abovementioned FFT computation can be performed at the query phase,
as it takes $\Ot(p)$ time and $p$ is chosen to be small enough 
(we later optimize $p$ to be around $n^{1.5+\eps}$).
For every target $c$, the value $\cg{c\bmod p}$ is
clearly an overcount of $\cnt_{A'\times B'}(c)$, as it counts the entire
class of $(c\bmod p)$ within $A'\times B'$, including pairs whose sum differs
from $c$; these pairs are called the \emph{false positives} of $c$. In
expectation, each target has $O(n^{2}/p)$ false positives.
We also define the \emph{class} of residue $r$
to be the set of pairs in $A\times B$ whose sum is congruent to $r$.

The second building block is the data structure of Fiat and Naor~\cite{FN}
for function-inversion (abbreviated FN), which was used also in~\cite{KKMV26}.
Given a constant time oracle for a function $f$ over a finite domain, 
one can build a (randomized) data structure that, when queried for a value $y$, 
reports a preimage of $y$ (if it exists) within time bound $T$. 
The FN-structure offers a tradeoff between space and query time $T$, in which the space also grows with the collision probability of $f$.

Previous work by~\cite{KPS,KKMV26} 
classify every target $c$ as \emph{heavy} or \emph{light},
depending on whether the number of pairs in $A\times B$ summing to it exceeds a threshold $\Ot(n^{\delta})$ for a parameter $\delta\in(0,1)$. 
The number of heavy targets is thus at most $\Ot(n^{2-\delta})$.
While prior work stores the set of heavy targets, 
we use instead an approximation, the set of \emph{marked} targets:
with high probability it contains every heavy target, it may contain some light targets as well,
and unlike the heavy set, it has a compact representation of $\Ot(n^{1-\delta})$ words.
Only this representation is stored, and the query phase recomputes the set of marked targets from it.
Specifically, the preprocessing phase draws a random sample $R\subseteq A$ of $\Ot(n^{1-\delta})$ elements,
and a target is marked if it belongs to $R+B$ (where $R+B$ denotes the sumset addition $\{r+b\ |\ r\in R,b\in B\}$).

Our preprocessing phase builds a function-inversion data structure for
\[
  g:A\times B\to[p], 
  \qquad 
  g:(a,b)\mapsto (a+b \bmod p),
\]
restricted to the \emph{unmarked pairs} (pairs whose sum is unmarked).
The preimages of a residue $r$ are then precisely
all the unmarked pairs in $A\times B$ whose sum is congruent to $r$.
Notice that an inversion recovers only a single preimage of this $r$, 
whereas our intended use requires recovering all its preimages,
this is the first challenge we address further below.

Our query algorithm computes the answer, for each target $c\in C'$,
using the recovered preimages for its residue $r=(c\bmod p)$, as follows. 
We split the pairs counted by $\cg{r}$ according to their (true) sum,
which is congruent to $c$ and thus has three possibilities:
it can be $c$, or a marked value other than $c$,
or an unmarked value other than $c$.
That is, 
\begin{equation}
\label{eq:cnt}    
\small
  \cg{r}
  = \cnt_{A'\times B'}(c)
  \;+\;|\{\text{pairs with marked sum}\ne c\}|
  \;+\;|\{\text{pairs with unmarked sum}\ne c\}|.
\end{equation}
The last term is computed by recovering the preimages of $r$ 
and scanning them for pairs that lie in $A'\times B'$ and their (true) sum differs from $c$.
Whenever the middle term is known to be zero,
the algorithm can find $\cnt_{A'\times B'}(c)$ from the other terms in \eqref{eq:cnt}, 
and compute the answer for $c$, which is \textsc{yes} iff $\cnt_{A'\times B'}(c)>0$. 
Recovery that silently misses a preimage would undercount the last term,
which is the second challenge we address below,
and ensuring that the middle term vanishes is the third challenge.

We now address the three challenges mentioned above.
First, an FN data structure returns one preimage, 
whereas our use above requires all the preimages.
Therefore, similarly to \cite{KKMV26}, we partition $A$ at random into $\Ot(n^{\delta})$ sets 
and build a separate inversion structure for each set, 
and furthermore repeat the random partitioning $\Theta(\log n)$ times. 
Since an unmarked sum is light and thus has $\Ot(n^{\delta})$ preimages,
for a given preimage $(a,b)$, with high probability 
in at least one of the random partitions the element $a$ is the only 
preimage in its set, in which case the inversion structure will find it.
Using this preimage we can recover all the pairs $(a, \hat{b})$ with $\hat b = b \pmod p$,
by simply storing the set $B$ bucketed by residue mod $p$.
One technical point remains: the FN data structure only inverts a self-map 
(the range is the same as the domain), whereas $g$ above maps pairs to residues; 
our construction closes this gap by extending $g$ to a self-map.

A second challenge is that an inversion structure might silently fail to report some pair, 
in which case the last term in \eqref{eq:cnt} would be undercounted.
By construction, the recovered set of pairs is a subset of the unmarked pairs in this residue class,
and thus the value of $\cnt_{A'\times B'}(c)$ obtained from \eqref{eq:cnt}
is never smaller than the true value.
Our algorithm therefore executes the above construction $\Theta(\log n)$ times with independently chosen primes,
computes this value in each execution,
and answers \textsc{yes} iff the smallest value is positive.
The answer is thus correct whenever at least one execution recovers all the unmarked pairs in its class,
and a target that has a pair summing to it is always answered \textsc{yes}.
A check described in Section~\ref{sec:cert} certifies, for each target, whether the answer is correct.

A third challenge is to ensure that the middle term in~\eqref{eq:cnt} is zero.
In general, it might be non-zero, 
in which case the residue class must contain a marked value other than $c$.
However, we show that for any fixed target $c$, a single execution likely places no other marked value in the same class as $c$; we call such an execution a \emph{clean run} for $c$.
In a clean run, every pair counted by $\cg{r}$ whose sum differs from $c$ has an unmarked sum,
and hence the value calculated in a clean run that recovers all the unmarked pairs in the class is exact.
With high probability, every target has many clean runs among the $\Theta(\log n)$ executions,
and at least one of them recovers all the unmarked pairs in its class.

The foregoing discussion of the three challenges establishes correctness,
and it remains to control the space of the inversion structures. 
The size of an FN data structure depends on (the query-time bound and)
the collision probability of the inverted function,
that is, on the number of pairs that are mapped to the same residue.
Here, the preimage of a residue is the set of unmarked pairs in a whole class, and its size,
which we call a \emph{load}, is determined by the selected prime and the partition of $A$.
The loads, and with them the sizes of the FN structures, can be of different values, so we bound their total size rather than each one separately.
The preprocessing phase redraws the prime and partition when the total load is too large.

The construction draws all its random bits during the preprocessing phase, 
and the query procedure is deterministic. 
In the query phase we cap the work spent on each attempt to ensure a worst-case running time.
The preprocessing phase succeeds with high probability, 
and upon success, the data structure answers all possible queries correctly. 
The guarantees therefore do not depend on how the queries are chosen, and hold in particular when they are chosen adaptively.

We conclude with a comparison to~\cite{KKMV26}.
There, heavy and light targets are answered by two separate methods.
Light targets are answered by inverting the exact sum $(a,b)\mapsto a+b$,
whereas for heavy targets, the false positives are computed
in the preprocessing phase and stored as lists, which dominate the space bound.
Our algorithm uses function inversion for a different purpose.
It inverts the residue map $(a,b)\mapsto (a+b\bmod p)$ on $A\times B$,
and thereby recovers at query time the pairs of the queried residue class,
which are exactly the pairs that~\cite{KKMV26} store.
This removes the stored lists of false positives and unifies the query algorithm,
so that heavy and light targets are answered in the same way.
Likewise, the set of heavy targets, which is stored explicitly in~\cite{KKMV26},
is replaced by the marked set, which is recomputed at query time from the sample $R$.
In addition, in~\cite{KKMV26} each FN data structure is bounded separately,
through a partition that makes every preimage small;
we bound instead the aggregated space complexity of all the FN structures.

\subsection{Related work}\label{sec:related}

Another related variant is 3SUM-indexing~\cite{DV,GKLP},
where the task is to preprocess sets $A,B$ of size $n$,
and then given as query a single target $c$,
determine whether there exists a pair $(a,b)\in A\times B$ such that $a+b=c$. 
Goldstein, Kopelowitz, Lewenstein, and Porat~\cite{GKLP}
introduced 3SUM-indexing as a basis for conditional lower bounds in data structures, 
and conjectured that sublinear query time requires quadratic space. 
This conjecture was refuted independently by~\cite{KP19} and~\cite{GGH}, 
both using the Fiat--Naor function-inversion scheme~\cite{FN}. 
Later improvements to the function-inversion
tradeoff apply in the regime where the inversion time is large,
comparable to the domain size~\cite{GGPS,DG26}. Unconditional lower
bounds for 3SUM-indexing were proved in~\cite{GGH} and subsequently strengthened, 
including for adaptive data structures~\cite{ChL}. 
Function inversion is also studied as a general
cryptographic data-structure problem; see~\cite{DBLP:conf/tcc/Corrigan-GibbsK19} for an overview of known
constructions, barriers, and open questions.

\section{Preliminaries}\label{sec:prelim}

We work in the word RAM with words of
$O(\log n)$ bits. All input integers are bounded in absolute value by
a polynomial in $n$. 
Write $[n]=\{0,\dots,n-1\}$. 
For integers
$u,v$ and a prime $p$, $u\equiv_{p}v$ denotes congruence mod $p$. For
$c\in\mathbb{Z}$, $\cnt(c)=|\{(a,b)\in A\times B:\ a+b=c\}|$, and for
sets $X\subseteq A$, $Y\subseteq B$,
$\cnt_{X\times Y}(c)=|\{(a,b)\in X\times Y:\ a+b=c\}|$.

The residue-class counts of a sumset are computed by a single FFT.

\begin{lemma}[\cite{KPS}, Lem.~2.1]\label{lem:fft}
Given sets $A,B$ of at most $n$ integers and a modulus $p$, the $p$
residue-class counts of the multiset $A+B\bmod p$ can be computed in
$O(n+p\log p)$ time.
\end{lemma}

The next lemma bounds the collision probability of two fixed integers
modulo a random prime.

\begin{lemma}[\cite{KPS}, Lem.~2.2]\label{lem:fp}
Let $u\neq v$ be integers bounded by $n^{O(1)}$ in absolute value, and
let $p$ be a uniformly random prime in $[R,2R]$, for any
$2\le R\le n^{O(1)}$. Then $\Pr[u\equiv_{p}v]=O(\log n/R)$.
\end{lemma}

We utilize the classic data structure for function inversion by Fiat and Naor~\cite{FN}.

\begin{lemma}[Fiat--Naor~\cite{FN}]\label{lem:fn}
Let $f:[N]\to[N]$ be a self-map. For any $T\le N$ there is a randomized construction producing a data structure of size $S=\Ot\bigl(\sqrt{N^{3}q(f)/T}\bigr)$ in
$\Ot(N+S)$ time, such that w.h.p over the construction randomness, on every query $y$ returns \emph{some} $x\in f^{-1}(y)$
(or $\bot$ if $f^{-1}(y) = \emptyset$) in $\Ot(T)$ time. Here
$q(f)=\Pr_{x,x'}[f(x)=f(x')]$ is the collision probability of $f$. 
\end{lemma}

\section{The algorithm}\label{sec:constr}

The data structure is built $J=\Theta(\log n)$ times independently.
Each copy is a \emph{run}, and run $j$ draws its own uniformly random
prime $p_{j}\in[n^{2-\delta}\log^{3}n,2n^{2-\delta}\log^{3}n)$ and its own random partitions of $A$.
We write $p$ for the prime of a generic run and omit the run index except where runs interact.
The construction uses two parameters, $\delta\in[0,1)$ (the exponent of the heavy threshold)
and $t$ (the exponent of the inversion time), and we set
$L:=6n^{\delta}\log^{2}n$ (the number of parts in each partition),
assuming for simplicity that $L$ divides $n$.
For $c\in A+B$, we say that $c$ is \emph{heavy} if $\cnt(c)\ge n^{\delta}\log^{2}n$ and \emph{light} otherwise.
The values of $\delta$ and $t$ are optimized in Section~\ref{sec:space}.

\subsection{Preprocessing}\label{sec:pre}

On input $(A,B)$, construct the following data structures, shown by nesting level in
Figure~\ref{fig:layout}.

\begin{itemize}[itemsep=2pt]
\item \emph{The marked set.} Store a subset $R\subseteq A$, drawn uniformly at random, of size
$3n^{1-\delta}$. Let $\Ct:=R+B=\{a+b:\ a\in R,\ b\in B\}$.
A value $c$ is \emph{marked} if $c\in\Ct$ and \emph{unmarked} otherwise,
and a pair $(a,b)\in A\times B$ is denoted marked or unmarked according to its sum $a+b$.

Assume for now that every heavy value is marked. Algorithm~\ref{alg:preprocess} verifies this and fails otherwise, so it is guaranteed after a successful preprocessing.

\item \emph{Per run: partitions and FN structures.}

The set $A$ is randomly partitioned into $L$  equal sized sets, this is repeated 
$K=\Theta(\log n)$ times independently.
For partition $k\in[K]$ and set $i\in[L]$, $A^{(k)}_i$
denotes the $i$-th set. 
Let
\[
  M_r=\{(a,b)\in A\times B:\ a+b\equiv_{p}r,\
  a+b\notin\Ct\}
\]
be the set of unmarked pairs of residue class $r$.
We construct $KL$ corresponding FN structures inverting the function:
\[
  \smallmap^{(k)}_i(a,b)=(a+b)\bmod p,
  \qquad\text{on the unmarked pairs of } A^{(k)}_i\times B ,
\]
whose preimage at $r$ is exactly $M_r\cap(A^{(k)}_i\times B)$.
To apply Lemma~\ref{lem:fn}, $\smallmap^{(k)}_i$ is extended to a self-map
$\extendedmap^{(k)}_i$ of $[2p]$, by diverting marked pairs and dummies into $\{p,p+1,\ldots,2p-1\}$.

The precise
definition and its properties are given in Section~\ref{sec:anpre}. The structure of Lemma~\ref{lem:fn}, which we
call an \emph{FN structure}, is constructed for $\extendedmap^{(k)}_i$, for every
$k\in[K]$ and $i\in[L]$, with query-time parameter $T=n^{t}$. 
The query procedure of an FN structure evaluates $\extendedmap^{(k)}_i$,
which requires membership tests in $\Ct$, and therefore the query phase
recomputes $\Ct$ before querying any FN structure.
\end{itemize}

After drawing a run, compute the loads
\[
  d^{(k)}_{i,r} := \big| \{(a,b)\in A^{(k)}_i\times B:\ a+b\notin\Ct,\
  a+b\equiv_{p}r\} \big|
\]
and check, for every
partition $k$,
\begin{equation*}\label{eq:loadcheck}
  \textstyle\sum_{i,r} \big(d^{(k)}_{i,r} \big)^{2}\ \le\ 2n^2\log^2 n .  
\end{equation*}
We call this test the \emph{load check}.
If the check fails for any $k$, redraw the run (prime and partitions). The loop is capped at $\log n$ attempts, and exhausting
them counts as a failed preprocessing.
In what follows, we omit the superscript $k$ whenever it is clear from context.

The stored data structure consists of the sets $A$, $B$ and $R$,
and for every run its prime, its $K$ partitions and its $KL$ FN structures.

\begin{figure}[t]
\centering
\definecolor{heavycol}{RGB}{95,130,80}   
\definecolor{arrcol}{RGB}{70,105,120}    
\begin{tikzpicture}[
  x=1cm,y=1cm,
  bx/.style={draw=black!60},
  gh/.style={draw=black!30},
  card/.style={bx,fill=blue!4},
  cardb/.style={gh,fill=blue!4},
  card2/.style={bx,fill=blue!9},
  card2b/.style={gh,fill=blue!9},
  brc/.style={decorate,line width=0.45pt,black!45},
  hdr/.style={font=\scriptsize,anchor=west},
  loop/.style={font=\scriptsize,anchor=south west,inner sep=1pt},
  tbl/.style={draw=black!55,fill=black!8,font=\scriptsize,inner sep=3pt},
  hset/.style={draw=heavycol!75,line width=0.8pt,fill=heavycol!14,
               font=\scriptsize,inner sep=3pt},
  fnb/.style={draw=black!55,fill=orange!15,font=\tiny,inner sep=2pt,
              minimum width=0.55cm,minimum height=0.36cm},
  dots/.style={font=\scriptsize,black!70},
  ar/.style={->,black!55,line width=0.4pt},
  slab/.style={font=\scriptsize}]

\draw[bx,thick] (0,2.38) rectangle (5.95,7.95);
\node[hset,anchor=west] at (0.45,7.60) {sample $R$};

\node[loop] at (0.45,6.92) {\textbf{for} $j\in[J]$};
\draw[cardb] (0.75,2.72) rectangle (5.55,6.60);
\draw[cardb] (0.60,2.87) rectangle (5.40,6.75);
\draw[card]  (0.45,3.02) rectangle (5.25,6.90);
\node[hdr] at (0.55,6.65) {run $j=1$:\quad prime $p_{j}$};

\node[loop] at (0.85,5.87) {\textbf{for} $k\in[K]$};
\draw[card2b] (1.15,3.17) rectangle (5.05,5.55);
\draw[card2b] (1.00,3.32) rectangle (4.90,5.70);
\draw[card2]  (0.85,3.47) rectangle (4.75,5.85);
\node[hdr] at (0.92,5.60) {partition $k=1$};

\node[slab,anchor=east] at (1.26,4.60) {$A$};
\draw[bx,fill=arrcol!18] (1.32,4.41) rectangle (4.42,4.79);
\draw[arrcol!55,line width=0.3pt]
  (1.60,4.43)--(1.60,4.77) (1.88,4.43)--(1.88,4.77)
  (2.44,4.43)--(2.44,4.77) (2.72,4.43)--(2.72,4.77)
  (3.86,4.43)--(3.86,4.77) (4.14,4.43)--(4.14,4.77);
\foreach \x in {2.16,3.00,3.58}{\draw[black!85,line width=0.9pt] (\x,4.38)--(\x,4.82);}
\node[dots] at (3.29,4.60) {$\cdots$};
\draw[brc,decoration={brace,amplitude=5pt}] (1.32,4.89) -- (4.42,4.89);
\node[slab] at (2.87,5.21) {$L$ parts};

\draw[ar] (1.74,4.37) -- (1.74,4.09);
\draw[ar] (2.58,4.37) -- (2.58,4.09);
\draw[ar] (4.00,4.37) -- (4.00,4.09);
\node[fnb] at (1.74,3.87) {FN$_1$};
\node[fnb] at (2.58,3.87) {FN$_2$};
\node[dots] at (3.29,3.87) {$\cdots$};
\node[fnb] at (4.00,3.87) {FN$_L$};
\end{tikzpicture}
\caption{The stored data structures by nesting level. $J=\Theta(\log n)$ runs,
$K=\Theta(\log n)$ partitions per run, $L = \Ot(n^\delta)$ FN-structures per
partition.}
\label{fig:layout}
\end{figure}
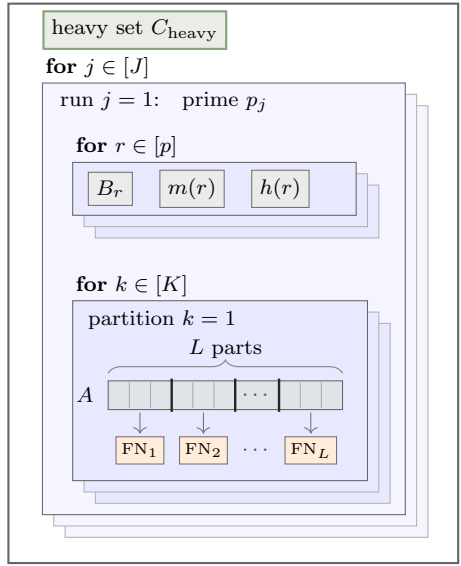

\SetKwRepeat{RepeatTimes}{repeat $\log n$ times}{until}
\begin{algorithm}[H]
\SetKwInOut{KwIn}{Input}
\SetKwInOut{KwOut}{Output}
\KwIn{Sets $A,B$ of $n$ integers}
\KwOut{The stored objects of Section~\ref{sec:pre}, or \textsc{failure}}
draw a uniformly random $R\subseteq A$ of size $3n^{1-\delta}$, and let $\Ct\gets R+B$\;
\lIf{some $c$ with $\cnt(c)\ge n^{\delta}\log^{2}n$ has $c\notin\Ct$}{\Return{\textsc{failure}}}
\For{$j \in [J]$}{
  \RepeatTimes{every partition $k\in [K]$, satisfies $\sum_{i,r}d_{i,r}^{2}\le 2n^{2}\log^{2}n$,
               or \Return{\textsc{failure}}\nllabel{step:loadcheck}}{
    draw a prime $p_{j}\in[n^{2-\delta}\log^{3}n,2n^{2-\delta}\log^{3}n)$, and $K$ balanced partitions of $A$ into $L$ sets\;
    compute the loads $d^{(k)}_{i,r}$ for all $k\in[K]$, $i\in[L]$, $r\in[p_{j}]$\;
  }
  construct FN-structure for $\extendedmap^{(k)}_i$, for all $k\in[K]$ and $i\in[L]$, with query time $T=n^{t}$\;
}
\caption{Preprocess$(A,B)$}
\label{alg:preprocess}
\end{algorithm}

\subsection{The subroutine Recover}\label{sec:rec}

$\textsc{Recover}$ is a subroutine of the query procedure, which invokes it once per run for
each target. $\textsc{Recover}(r)$ attempts to output
$M_r$, the unmarked pairs of residue class $r$, within one run.
Its output is always a subset of $M_r$, and the analysis shows that it is likely to be all of $M_r$.
Figure~\ref{fig:enum} shows the control flow. The two
steps:
\begin{enumerate}[itemsep=1pt]
\item Query each of the $KL$ FN-structures at the value $r$. Each
returned domain point is verified by evaluating $G^{(k)}_i$ and
discarded unless its value is $r$. Set
\[
  \widehat{A}=\{a\in A:\ \text{some verified returned pair has first coordinate } a\}.
\]
\item For each $a\in\widehat{A}$, read the set $B_{(r-a)\bmod p}$, where $B_{(r-a)\bmod p}=\{b\in B:\ b\equiv_{p}{r-a}\}$ is calculated in the query phase, and output the pairs $(a,b)$ with $a+b\notin\Ct$. If the total number of
values read exceeds $L$, stop and return the pairs output so far.
\end{enumerate}

\begin{figure}[t]
\centering
\begin{tikzpicture}[
  stp/.style={draw, rounded corners=1pt, align=center, font=\small,
               inner sep=4pt},
  arr/.style={->, font=\footnotesize}]
\node[stp] (d) at (0,0) {\emph{discover}\\ $KL$ FN-calls at $r$\\
  $\to$ coordinates $\widehat{A}$};
\node[stp, right=30pt of d] (e)
{\emph{expand}\\
for every $a \in \widehat{A}$, read the set $B_{(r-a)\bmod p}$:\\
output the unmarked pairs (at most $L$ reads)};
\node[stp, right=30pt of e] (o)
{\emph{output}\\ $\mathit{out}\subseteq M_r$};
\draw[arr] (d) -- (e);
\draw[arr] (e) -- (o);
\end{tikzpicture}
\caption{One execution of $\textsc{Recover}(r)$. The output is
always a subset of $M_r$, and the count computed from it is never
below the true count.}
\label{fig:enum}
\end{figure}

\begin{algorithm}[H]
\SetKwInOut{KwIn}{Input}
\SetKwInOut{KwOut}{Output}
\KwIn{Residue $r\in[p]$ 
  \tcp{within one run}}
\KwOut{A subset $\mathit{out}$ of $M_{r}$}
$\widehat{A}\gets\emptyset$;\quad $\mathit{out}\gets\emptyset$\;
\tcp{discover}
\For{$k\in[K]$ and $i\in[L]$}{
  $x\gets$ the answer of the FN structure of $\extendedmap^{(k)}_i$, queried at $r$\;
  \lIf{$x\ne\bot$ and $\extendedmap^{(k)}_i(x)=r$}{add the first coordinate of $x$ to $\widehat{A}$}
}
\tcp{expand}
\For{$a\in\widehat{A}$}{
  \For{$b\in B_{(r-a)\bmod p}$}{
    \lIf{more than $L$ values of $b$ were read so far\nllabel{step:cap}}{\Return{$\mathit{out}$}}
    \lIf{$a+b\notin\Ct$}{add $(a,b)$ to $\mathit{out}$}
  }
}
\Return{$\mathit{out}$}\;
\caption{Recover$(r)$}
\label{alg:Recovery}
\end{algorithm}

\subsection{Query}\label{sec:query}

Compute the marked set $\Ct=R+B$, stored as a static dictionary,
and for each run $j$ the sets $B_{r}:=\{b\in B:\ b\equiv_{p_{j}}r\}$ for all $r\in[p_{j}]$,
stored in an array of size $p_{j}$.
Compute, for each run $j$ and all residues $r$ at once, the counts
$\cgj{r}{j} := |\{(a,b)\in A'\times B':\ a+b\equiv_{p_{j}}r\}|$,
by Lemma~\ref{lem:fft}; we omit the run index $j$ whenever it is clear from
context. Each target $c\in C'$ is then
processed independently, as follows. 
For every run $j$, compute
$r=c\bmod p_{j}$, execute that run's
$\textsc{Recover}(r)$, and compute $z_j(c)$ from its output $\mathit{out}$ :
\[
  z_{j}(c) := \cgj{r}{j} - \bigl|\{(a,b)\in\mathit{out}\cap(A'\times B'):\ a+b\neq c\}\bigr| .
\]
Answer \textsc{yes} for $c$ iff $\min_{j\in[J]} z_{j}(c)>0$.
Section~\ref{sec:cert} describes a check that certifies, for a target $c$ and a run $j$, that $z_{j}(c)$ is exact.

\begin{algorithm}[H]
\SetKwInOut{KwIn}{Input}
\SetKwInOut{KwOut}{Output}
\KwIn{Subsets $A'\subseteq A$, $B'\subseteq B$, and a set $C'$ of targets}
\KwOut{A \textsc{yes}/\textsc{no} answer for every $c\in C'$}
compute $\Ct\gets R+B$, and the sets $B_{r}$ for every $j\in[J]$ and all residues $r\in[p_{j}]$\;
compute $\cgj{r}{j}$ for every $j\in[J]$ and all residues $r$ (Lemma~\ref{lem:fft})\;
\For{$c\in C'$}{
  \For{$j\in[J]$}{
    $r\gets c\bmod p_{j}$\;
    $\mathit{out}\gets$ run $j$'s $\textsc{Recover}(r)$\;
    $z_{j}(c)\gets \cgj{r}{j} - |\{(a,b)\in\mathit{out}\cap(A'\times B'):\ a+b\neq c\}|$\;
  }
  answer \textsc{yes} for $c$ iff $\min_{j\in[J]} z_{j}(c)>0$\;
}
\caption{Query$(A',B',C')$}
\label{alg:query}
\end{algorithm}

\section{Analysis}\label{sec:analysis}

\subsection{Preprocessing}\label{sec:anpre}

We first show that the marked set covers the heavy values.

\begin{lemma}\label{lem:cover}
With high probability over the choice of $R$, every heavy value is marked.
\end{lemma}

\begin{proof}

Fix a heavy value $c$. There are $\cnt(c)\ge n^{\delta}\log^{2}n$ elements $a\in A$ with $c-a\in B$,
and $c$ is marked iff $R$ contains at least one of them.
Since $R$ is a uniformly random subset of $A$ of size $|R|$, the probability that it contains none of them is at most
\[
  (1-\cnt(c)/n)^{|R|}\ \le\ e^{-|R|\cnt(c)/n}\ \le\ e^{-3\log^{2}n}\ \le\ n^{-3\log{n}} .
\]
A union bound over the at most $n^{2}$ values in $A+B$ completes the proof.
\end{proof}

By construction, $|\Ct|\le|R|\cdot|B|\le 3n^{2-\delta}$, and after a successful preprocessing
every unmarked value is light, that is, $\cnt(s)<n^{\delta}\log^{2}n$ for every unmarked value $s$.

To apply Lemma~\ref{lem:fn} we extend $g^{(k)}_i$ to a self-map of $[2p]$ as follows.
The $n^{2}/L\le p$ pairs of $A^{(k)}_i\times B$ are identified with the first $n^{2}/L$ points of $[2p]$ via the lexicographic order,  and call the remaining points \emph{dummies}.
Write $(a_x,b_x)$ for the pair identified with
$x$ and set
\[
  \extendedmap^{(k)}_i(x)=
  \begin{cases}
    (a_x+b_x)\bmod p, & x<n^{2}/L \text{ and } a_x+b_x\notin\Ct,\\[2pt]
    p+(x\bmod p),     & \text{otherwise.}
  \end{cases}
\]
The first branch applies to the unmarked pairs and reproduces
$\smallmap^{(k)}_i$; the second diverts the marked pairs and the
dummies into the label range $[p,2p)$, which residue queries never
touch. Note that for every residue $r\in[p]$, the preimage of $r$ is
exactly $M_r\cap(A^{(k)}_i\times B)$, whose size is $d_{i,r}$.
Figure~\ref{fig:pad} shows which part of the domain occupies which
part of the range.

\begin{figure}[H]
\centering
\begin{tikzpicture}[font=\small]
\draw (0,2.1) rectangle (3.4,2.7);
\draw (3.4,2.1) rectangle (10.2,2.7);
\node at (1.7,2.4) {pairs\ \ $[n^{2}/L]$};
\node at (6.8,2.4) {dummies};
\draw (0,0) rectangle (5.1,0.6);
\draw (5.1,0) rectangle (10.2,0.6);
\node at (2.55,0.3) {residues $[p]$};
\node at (7.65,0.3) {labels $[p,2p)$};
\draw[->] (1.2,2.05) -- node[left,font=\footnotesize,pos=0.5]
  {unmarked: $(a_x+b_x)\bmod p$\ } (1.2,0.65);
\draw[->] (2.7,2.05) -- node[below right,font=\footnotesize,pos=0.35]
  {otherwise: $p+(x\bmod p)$} (6.4,0.65);
\draw[->] (7.6,2.05) -- (7.9,0.65);
\end{tikzpicture}
\caption{The padded self-map $G^{(k)}_i:[2p]\to[2p]$. The
first $n^{2}/L\le p$ domain points are the pairs of
$A^{(k)}_i\times B$, the rest are dummies. Unmarked pairs land on residues. Marked pairs and
dummies are diverted to the label range $[p,2p)$, which residue
queries never touch. The preimage of a residue is therefore exactly
the unmarked pairs of the corresponding class.}
\label{fig:pad}
\end{figure}
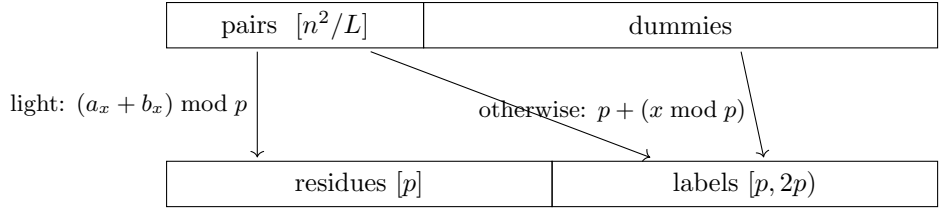

We now bound the collision probability of every map $G^{(k)}_i$ using the loads $d_{i,r}$.

\begin{claim}\label{claim:pad}
For every $i \in [L]$ and $k \in [K]$, it holds that, $q\bigl(G^{(k)}_i\bigr)\le 1/p + 
\bigl(\sum_{r \in [p]}d_{i,r}^{2} /4p^2 \bigr)$.
\end{claim}

\begin{proof}
Denote the maps $\extendedmap^{(k)}_i$ by $f$, then,
\[
q(f) = \sum_{r \in [2p]} \Pr[f(x)=r]  \Pr[f(x')=r] = \sum_{r \in [2p]} \Big(\frac{|f^{-1}(r)|}{2p}\Big)^2.
\]
The claim is obtained as for $r \in [p]$, $f^{-1}(r)$ is exactly $d_{i,r}$ and for $r \not \in [p]$, $f^{-1}(r)$ has at most two preimages (one dummy and one marked pair).
\end{proof}

The next lemma controls the loads, and with them the load check (line~\ref{step:loadcheck} of Algorithm~\ref{alg:preprocess}).

\begin{lemma}\label{lem:loadexp}
For each partition,
$\mathbb{E}\bigl[\sum_{i,r}d_{i,r}^{2}\bigr] \leq 2n^2$, the
expectation is over the prime and the partition jointly.
\end{lemma}

\begin{proof}
$\sum_{i,r}d_{i,r}^2$ is the number of ordered pairs $(a,b),(a',b')$ such that $a,a'$ belong to the same set, $a+b\equiv_{p}a'+b'$ and that sum is unmarked.
We use the following facts.
\begin{enumerate}[itemsep=1pt,label=(\alph*)]
\item If $a \neq a'$ then the probability both $a,a'$ belong to the same set is at most $2/L$.
\item If $a+b \neq a'+b'$ then the probability both are congruent mod $p$ is $O(1/(n^{2-\delta}\log^{2}n))$ (Lemma~\ref{lem:fp}).
\end{enumerate}
There are four cases.
\begin{enumerate}[itemsep=1pt]
\item $a=a'$ and $a+b=a'+b'$: then $b=b'$, there are at most  $n^{2}$ such ordered pairs.
\item $a\neq a'$ and $a+b=a'+b'$: The probability is at most $2/L$ by
(a). The two pairs sum to a common unmarked value $s$, and
for fixed $s$ the number of ordered pairs is $\cnt(s)^{2}$. 
Recall that an unmarked $s$ is light, hence $\cnt(s) \leq n^\delta \log^{2}n$, the contribution is at most
\[
(2/L)\cdot\sum_{s\ \mathrm{unmarked}}\cnt(s)^{2}\le
(2/L) \cdot n^\delta \log^{2}n \cdot\sum_{s\ \mathrm{unmarked}}\cnt(s)\le
2n^{2+\delta}\log^{2}n/L \leq n^2/3.
\]
\item $a=a'$ and $a+b\ne a'+b'$: hence $b\equiv_{p}b'$ which occurs with probability $O(1/(n^{2-\delta}\log^{2}n))$ by (b).
Such an ordered pair is
determined by a triple $(a,b,b')$, of which there are at most $n^{3}$,
so the contribution is $O(n^{1+\delta}/\log^{2}n) \le n^2/3$.
\item $a\ne a'$ and $a+b\ne a'+b'$: 
The probability is $O(1/(n^{2-\delta}\log^{2}n))\cdot 2/L$ by (a) and (b). There are at most $n^{4}$ ordered pairs, so the contribution is
$O(n^{2+\delta}/\log^{2}n)/L \leq n^2 /3$.
\end{enumerate}
Summing the four cases, the total is at most $2n^2$.
\end{proof}

By Lemma~\ref{lem:loadexp} and Markov's inequality, a fixed partition fails the load check (line~\ref{step:loadcheck} of Algorithm~\ref{alg:preprocess}) with probability at most $1/\log^{2}n$. A union bound over the $K$ partitions
makes each attempt pass with probability $1-O(1/\log n)$, 
hence after $\log n$ attempts the loop succeeds w.h.p.
We next bound the number of clean runs for each target, where
a run is \emph{clean} for a target $c$ if the residue class of $c$
contains no marked value other than possibly $c$ itself, that is, if
every $c'\in\Ct$ with $c'\equiv_{p}c$ satisfies $c'=c$.

\begin{lemma}\label{lem:clean}
With high probability, every possible target has at least $J/2$ clean runs.
\end{lemma}

\begin{proof}
Fix a target $c$ and a run. For each $c'\in\Ct$ with $c'\ne c$,
Lemma~\ref{lem:fp} gives $\Pr[c'\equiv_{p}c]=O(1/(n^{2-\delta}\log^{2}n))$. Consequently,
by a union bound the expected number of marked values other than $c$ in
the residue class of $c$ is at most
\[
  |\Ct|\cdot O(1/(n^{2-\delta}\log^{2}n))\ =\
  O(1/\log^{2} n),
\]
using $|\Ct|\le 3n^{2-\delta}$. The run is not clean for $c$ if
and only if this count is at least $1$, so by Markov's inequality 
each run is clean with probability at least $1-O(1/\log^{2} n)$, independently across runs.
Hence the number of clean runs is distributed $\textit{Bin}(J,1-O(1/\log^{2} n))$.
The proof is completed using a Chernoff bound together with a union bound over the $n^{O(1)}$ possible targets.
\end{proof}

\subsection{Recovery}\label{sec:anrec}

Throughout this subsection we fix a target $c$ and a run, and write
$r=c\bmod p$.

\begin{observation}\label{cl:cert}
Every element of $\widehat{A}$ is a first coordinate of some pair of $M_r$, and $\textsc{Recover}(r)$'s output is a subset of
$M_r$.
\end{observation}

\begin{proof}
    A point enters $\widehat{A}$ only after $\extendedmap^{(k)}_i$ is evaluated on it and the value is confirmed to be $r \in [p]$; by the definition of $\extendedmap^{(k)}_i$, its preimage at $r$ is $M_r \cap (A^{(k)}_i \times B)$. Step 2 emits $(a,b)$ only when $b \in B_{(r-a)\bmod p}$ and $a + b \notin \Ct$, i.e.  only when $(a,b) \in M_r$.
\end{proof}

So the output is all of $M_r$ unless some first coordinate of $M_r$ fails
to enter $\widehat{A}$, or the reading cap stops the expansion (line~\ref{step:cap} of Algorithm~\ref{alg:Recovery}). Both events are controlled using a bound on $|M_r|$ that we establish first.

\begin{lemma}\label{lem:mass}
With probability $1-O(1/\log n)$ over the run's randomness,
$|M_r|\le 2n^{\delta}\log^{2}n$.
\end{lemma}

\begin{proof}
Split $M_r$ according to whether a pair sums to $c$ or not. Pairs
summing to $c$ contribute at most $n^{\delta}\log^{2}n$, since they
contribute nothing when $c$ is marked and at most
$\cnt(c)<n^{\delta}\log^{2}n$ when $c$ is unmarked, as an unmarked value is light. For the remaining
pairs, an unmarked value $s\ne c$ contributes $\cnt(s)$ pairs and only
when $s\equiv_{p}c$, which by Lemma~\ref{lem:fp} happens with
probability $O(1/(n^{2-\delta}\log^{2}n))$. As $\sum_{s}\cnt(s)\le n^{2}$,
the expected contribution is $O(n^{\delta}/\log^{2}n)$, so by Markov's
inequality it exceeds $n^{\delta}\log^{2}n$ with probability at most
$O(1/\log^{4}n)$. The lemma is obtained by combining both parts.
\end{proof}

We bound the first bad event that some first coordinate of $M_r$ is not in $\widehat{A}$.

\begin{lemma}\label{lem:disc}
Condition on $|M_r|\le 2n^{\delta}\log^{2}n$. Then, w.h.p., every first
coordinate occurring in $M_r$ enters $\widehat{A}$.
\end{lemma}

\begin{proof}
Fix a first coordinate $a$ occurring in $M_r$ and a partition $k$, and denote by $i \in [L]$ the set $a$ belongs to. Say that $a$ is \emph{isolated} in $k$ if no other first coordinate occurring in $M_r$ lies in $A^{(k)}_i$. If $a$ is isolated, then, the preimage of $r$ under $G^{(k)}_i$ is
$M_r\cap(A^{(k)}_i\times B)$ and $a$ enters $\widehat{A}$.

It remains to control the probability that $a$ is isolated in some
partition. A fixed other coordinate lies in $a$'s set with probability
at most $2/L$, and by Observation~\ref{cl:cert} there are at most $|M_r| \leq 2n^\delta \log^{2}n $ first coordinates occurring in $M_r$, so by a union bound, $a$ is isolated in a
fixed partition with probability at least
$1-4n^{\delta}\log^{2}n/L \geq 1/3$. The $K$ partitions are drawn
independently, so w.h.p $a$ is isolated in some partition. A union bound over the at most $n$ first coordinates completes the proof.
\end{proof}

We then bound the second bad event that more than $L$ entries are read.

\begin{lemma}\label{lem:aband}
In a clean run with $|M_r|\le 2n^{\delta}\log^{2}n$, the expansion reads at
most $L$ entries.
\end{lemma}

\begin{proof}
For $a\in\widehat{A}$ the expansion reads the set
$B_{(r-a)\bmod p}$, and every value $b$ of that set gives a pair
$(a,b)$ of residue class $r$. If $a+b$ is unmarked then the number of values read is at most $|M_r|$.
If $a+b$ is marked, since the run is clean for $c$ we have $a+b=c$, so there are at most $|\widehat{A}|$ such pairs.
Since $|\widehat{A}|\le |M_r|$ by Observation~\ref{cl:cert}, the expansion
reads at most $2|M_r|\le 4n^{\delta}\log^{2}n\le L$ values.
\end{proof}

\subsection{Query}\label{sec:anquery}

We first show that the value of $z_j(c)$ is never too small, and that a clean run
whose $\textsc{Recover}(r)$ outputs all of $M_r$ yields an exact count.

\begin{claim}\label{cl:ans}
For every target $c$ and every run $j$, $z_{j}(c)\ge\cnt_{A'\times B'}(c)$.
If moreover the run is clean for $c$ and its $\textsc{Recover}(r)$ outputs all of $M_r$,
where $r=c\bmod p_{j}$, then $z_{j}(c)=\cnt_{A'\times B'}(c)$.
\end{claim}

\begin{proof}
Fix a target $c$ and a run, and write
$r=c\bmod p$. $\cg{r}$ counts all pairs of
$A'\times B'$ whose sum is equivalent to $r$ modulo $p$. Splitting by the exact sum:
\begin{multline*}
 \cg{r}
 =
 \underbrace{\cnt_{A'\times B'}(c)}_{\text{queried}}
 +
 \underbrace{\bigl|\{(a,b)\in M_r\cap(A'\times B'):\ a+b\ne c\}\bigr|}_{\text{other unmarked sums}}
 \\
 +
 \underbrace{\bigl|\{\text{pairs with \emph{other marked}
 sums}\}\bigr|}_{=\,0\ \text{in a clean run}} .
\end{multline*}
By Observation~\ref{cl:cert} the output of $\textsc{Recover}(r)$ is a subset of $M_r$,
so the quantity subtracted in $z_{j}(c)$ is at most the middle term,
and $z_{j}(c)\ge\cnt_{A'\times B'}(c)$.
Pairs summing to $c$ itself are never subtracted, because of the condition $a+b\ne c$.
If the output is all of $M_r$, the subtracted quantity equals the middle term,
and if moreover the run is clean for $c$, then no marked value other than $c$ is congruent to $c$,
so the last term is zero. Hence $z_{j}(c)=\cnt_{A'\times B'}(c)$.
\end{proof}

Since the answer for $c$ is \textsc{yes} iff $\min_{j}z_{j}(c)>0$,
the claim shows that the answer is correct whenever some run is clean for $c$ and outputs all of $M_r$,
and that in any case a target with $\cnt_{A'\times B'}(c)>0$ is answered \textsc{yes}.
It remains to show that such runs occur.

\begin{lemma}\label{lem:enum}
Fix a target $c$ and a run that is clean for $c$ and passed
the load check. With probability $1-O(1/\log n)$ over the randomness of
that run, $\textsc{Recover}(r)$ outputs all of $M_r$.
\end{lemma}

\begin{proof}
By Lemma~\ref{lem:mass} we have $|M_r|\le 2n^{\delta}\log^{2}n$ except
with probability $O(1/\log n)$, and given that, by
Lemma~\ref{lem:disc}, w.h.p.\ every first coordinate occurring in
$M_r$ enters $\widehat{A}$. Assume both events hold which occur with probability $1-O(1/\log n)$.

We claim that every $(a,b)\in M_r$ is emitted. As $a+b\equiv_{p}r$ we
have $b\in B_{(r-a)\bmod p}$, and $a\in\widehat{A}$, so the expansion
reads that set. By Lemma~\ref{lem:aband} the cap does not stop it
before the set is read in full, since the run is clean for $c$; thus
$b$ is read, and as $a+b\notin\Ct$ the pair is emitted. Hence $M_r\subseteq\mathit{out}$, and by
Observation~\ref{cl:cert} we get $\mathit{out}=M_r$.
\end{proof}

\begin{corollary}\label{cor:enum}
With high probability, for every possible target $c$ there is a run that is clean for $c$
and whose $\textsc{Recover}(r)$ outputs all of $M_r$.
\end{corollary}

\begin{proof}
Fix a target $c$. W.h.p.\ at least $J/2$ runs are clean for $c$
(Lemma~\ref{lem:clean}). The runs are tried independently, so Lemma~\ref{lem:enum} together with a union bound over the
$n^{O(1)}$ possible targets completes the proof.
\end{proof}

\section{Proof of Theorem~\ref{thm:main}}\label{sec:proof}

\subsection{The time-space tradeoff}\label{sec:space}

\paragraph{Space complexity.}

The sets $A$, $B$ and $R$, and for each run its prime and its $K$ partitions, take $\Ot(n)$ space in total.
By Lemma~\ref{lem:fn} and Claim~\ref{claim:pad},
the structure for $G^{(k)}_i$ has size 
\[
S_i=\Ot\Bigl(\sqrt{p\textstyle\sum_r d_{i,r}^{2}/T}
  \;+\;p/\sqrt{T}\Bigr).
\]
The run passed the load check, so $\sum_{i}\sum_r d_{i,r}^{2}\le
2n^{2}\log^{2}n$,
furthermore, we have the inequality $\sum_{i\le L}\sqrt{x_i}\le\sqrt{L\sum_i x_i}$ (following Cauchy--Schwarz).
Recall $T=n^t$ and $Lp=\Ot(n^{2})$, we can then write
\[
  \sum_{i\in[L]} S_i\ \le\ \Ot\Bigl(\sqrt{Lp\sum_{i,r}d_{i,r}^{2}/T}\;+\;Lp/\sqrt{T}\Bigr) = \Ot \bigl( \frac{n^2}{\sqrt{T}} \bigr)
  = \Ot \bigl( n^{2-t/2}  \bigr).
\]
The total space complexity is  $\Ot \bigl( n+n^{2-t/2}  \bigr)$.

\paragraph{Query time complexity.}
The query time is deterministic and so the bound below holds in worst case.
Computing the marked set $\Ct=R+B$ takes $\Ot(|R|\cdot n)=\Ot(n^{2-\delta})$ time,
and computing the sets $B_{r}$ and the counts $\cg{r}$ takes $\Ot(p)$ time
per run by Lemma~\ref{lem:fft}, and $\Ot(p)=\Ot(n^{2-\delta})$ over the $J$ runs. Each
of the $O(n)$ targets then executes $\textsc{Recover}$ $J$
times, and one execution performs $KL$ inverter calls at $\Ot(T)$
each and reads at most $L$ values, for $\Ot(KLT)=\Ot(n^{\delta+t})$ time. Computing $z_{j}(c)$
from the output costs additional $O(|\mathit{out}|)=O(L)$
time. The total time complexity is $\Ot(n^{2-\delta}+n^{1+\delta+t})$.

Set $\delta=\tfrac12-\eps$ and $t=2\eps$. Then the query time is $\Ot(n^{3/2+\eps})$
and the space is $\Ot(n^{2-\eps})$.

\subsection{Failure and adaptive queries}\label{sec:adaptive}

The preprocessing succeeds w.h.p following Section~\ref{sec:anpre}, Lemma~\ref{lem:fn} and Corollary~\ref{cor:enum}.
On success, the data structure behaves deterministically: every
inverter answer is a fixed function of its query, so the
output of every execution of $\textsc{Recover}$ is a fixed function of the stored state and
the target, with $A'$ and $B'$ entering only the final membership
tests. Every possible target has a run that is clean for it and recovers all of $M_r$ (Corollary~\ref{cor:enum}), so
every answer is correct
(Claim~\ref{cl:ans}); and the running
time is bounded in the worst case (Section~\ref{sec:space}). These
guarantees hold for all queries simultaneously, hence are unaffected
by how the queries are chosen, and in particular by choosing them
adaptively as a function of earlier answers and running times. This
completes the proof of Theorem~\ref{thm:main}.

\subsection{Certified answers}\label{sec:cert}

Section~\ref{sec:adaptive}  guarantees that
if the preprocessing succeeds, then all answers are correct.
If the preprocessing does not succeed, then some answers can be wrong \textsc{yes} answers.
In this section we present an additional certification step that ensures 
the data structure never gives a wrong answer.

The query procedure never gives a wrong \textsc{no} answer (Claim~\ref{cl:ans}).
It can give a wrong \textsc{yes} answer for a target $c$ only in one case:
no run is clean for $c$ and outputs all of $M_r$.
After a successful preprocessing, this case does not occur (Corollary~\ref{cor:enum}).
However, the query procedure cannot tell whether a given  $z_{j}(c)$ is exact.
This subsection adds a check that detects this, for every target and every run.
With the check, every answer that the data structure gives is correct,
for every data structure that the preprocessing phase outputs.
If no run passes the check for a target, the data structure reports failure for this target.

\paragraph{The check.}
Fix a run and a target $c$, and let $r=c\bmod p$.
The class of $r$ is the set of all pairs $(a,b)\in A\times B$ with $a+b\equiv_{p}r$, marked or unmarked.
The set $M_r$ is the set of unmarked pairs in the class of $r$.
Let $\mathit{out}$ be the output of $\textsc{Recover}(r)$.
Define
\[
  \Phi(c) := \sum_{(a,b)\ \text{in the class of } r} (a+b-c)^{2},
  \qquad
  \widehat{\Phi}(c) := \sum_{(a,b)\in\mathit{out}} (a+b-c)^{2}.
\]
The sum $\widehat{\Phi}(c)$ is over all of $\mathit{out}$, not only over $\mathit{out}\cap(A'\times B')$.
The run is \emph{certified} for $c$ if $\Phi(c)=\widehat{\Phi}(c)$.

\begin{claim}\label{cl:certsound}
$\Phi(c)\ge\widehat{\Phi}(c)$.
Equality holds if and only if $\mathit{out}$ contains every pair in the class of $r$ whose sum differs from $c$.
If the run is certified for $c$, then $z_{j}(c)=\cnt_{A'\times B'}(c)$ for all $A'\subseteq A$ and $B'\subseteq B$.
\end{claim}

\begin{proof}
By Observation~\ref{cl:cert}, $\mathit{out}\subseteq M_r$, and $M_r$ is a subset of the class of $r$.
Hence $\Phi(c)-\widehat{\Phi}(c)$ is the sum of $(a+b-c)^{2}$ over the pairs of the class of $r$ that are not in $\mathit{out}$.
Every term of this sum is nonnegative.
A term is zero if and only if $a+b=c$.
This proves the first two statements.
For the third statement, recall that $\cg{r}$ is the number of pairs of $A'\times B'$ in the class of $r$.
The value $z_{j}(c)$ subtracts from $\cg{r}$ the number of pairs in $\mathit{out}\cap(A'\times B')$ whose sum differs from $c$.
If the run is certified, these pairs are all the pairs of $A'\times B'$ in the class of $r$ whose sum differs from $c$.
The pairs that remain are exactly the pairs of $A'\times B'$ that sum to $c$.
Hence $z_{j}(c)=\cnt_{A'\times B'}(c)$.
\end{proof}

\paragraph{Computing the check.}
The query procedure computes $\widehat{\Phi}(c)$ by a scan of $\mathit{out}$, in $O(|\mathit{out}|)$ time.
To compute $\Phi(c)$, expand $(a+b-c)^{2}=(a+b)^{2}-2c(a+b)+c^{2}$.
This gives
\[
  \Phi(c) = U_{2}(r) - 2c\,U_{1}(r) + c^{2}\,U_{0}(r),
  \qquad\text{where}\qquad
  U_{e}(r) := \sum_{(a,b)\ \text{in the class of } r} (a+b)^{e}
\]
for $e\in\{0,1,2\}$.
The arrays $U_{0},U_{1},U_{2}$ do not depend on the target.
The query procedure computes them once per run, for all residues $r\in[p]$, as follows.
For $u\in[p]$ and $e\in\{0,1,2\}$, let
\[
  S_{e}[u] := \sum_{a\in A:\ a\equiv_{p}u} a^{e},
  \qquad
  T_{e}[u] := \sum_{b\in B:\ b\equiv_{p}u} b^{e}.
\]
That is, $S_{0}[u]$ is the number of elements of $A$ with residue $u$,
$S_{1}[u]$ is their sum, and $S_{2}[u]$ is the sum of their squares, and similarly for $T_{e}[u]$ and $B$.
One scan of $A$ and one scan of $B$ compute these six arrays in $O(n+p)$ time.
The pairs $(a,b)$ in the class of $r$ with $a\equiv_{p}u$ are exactly the elements $a$ of residue $u$
combined with the elements $b$ of residue $(r-u)\bmod p$.
Hence, for every $r\in[p]$, with $v:=(r-u)\bmod p$ in every summand,
\begin{align*}
  U_{0}(r) &= \sum_{u\in[p]} S_{0}[u]\,T_{0}[v], \\
  U_{1}(r) &= \sum_{u\in[p]} \bigl( S_{1}[u]\,T_{0}[v] + S_{0}[u]\,T_{1}[v] \bigr), \\
  U_{2}(r) &= \sum_{u\in[p]} \bigl( S_{2}[u]\,T_{0}[v] + 2\,S_{1}[u]\,T_{1}[v] + S_{0}[u]\,T_{2}[v] \bigr),
\end{align*}
where the second line uses $\sum(a+b)=\sum a+\sum b$ and the third line uses $(a+b)^{2}=a^{2}+2ab+b^{2}$.
Each of the six sums $\sum_{u\in[p]}S_{i}[u]\,T_{j}[(r-u)\bmod p]$ is the cyclic convolution of $S_{i}$ and $T_{j}$ at $r$,
and it is computed for all $r\in[p]$ at once as follows.
The arrays $S_{i}$ and $T_{j}$ are the coefficient vectors of two polynomials of degree less than $p$.
The fast Fourier transform multiplies these polynomials in $O(p\log p)$ time, and the product has degree less than $2p$.
Adding the coefficient of $x^{r+p}$ to the coefficient of $x^{r}$, for every $r\in[p]$, gives the cyclic convolution.
All coefficients are integers of absolute value $n^{O(1)}$.

\paragraph{The query procedure with the check.}
In every run, the query procedure computes the arrays $U_{0},U_{1},U_{2}$ together with the counts $\cg{r}$.
For every target $c$ and every run $j$, after $\textsc{Recover}(r)$,
it computes $\widehat{\Phi}(c)$ and $\Phi(c)$ and tests whether they are equal.
If some run $j$ is certified for $c$, the answer for $c$ is \textsc{yes} if $z_{j}(c)>0$, and \textsc{no} otherwise.
By Claim~\ref{cl:certsound}, all certified runs give the same answer.
If no run is certified for $c$, the data structure reports failure for $c$.
The additional cost is $\Ot(p)$ time per run, and $O(|\mathit{out}|)=O(L)$ time per target and run, which does not change the overall query time. 

\begin{claim}\label{cl:certified}
With the check, every answer that the data structure gives is correct,
for every data structure that the preprocessing phase outputs.
After a successful preprocessing, the data structure does not report failure for any target of any query
\end{claim}

\begin{proof}
The data structure gives an answer only from a certified run.
By Claim~\ref{cl:certsound}, this answer is correct.
For the second statement, fix a possible target $c$.
By Corollary~\ref{cor:enum}, some run is clean for $c$ and its $\textsc{Recover}(r)$ outputs all of $M_r$.
In this run, consider a pair in the class of $r$ whose sum differs from $c$.
Its sum is unmarked, because in a clean run the class of $r$ contains no marked value other than $c$.
Hence the pair belongs to $M_r=\mathit{out}$.
By Claim~\ref{cl:certsound}, the run is certified for $c$.
\end{proof}

\section*{AI Disclosure}
All results, definitions, and proof strategies originate with the authors,
who used Claude (Anthropic) and ChatGPT (OpenAI) as a writing aid, both to
improve the clarity of author-written text and to write up proof details from
author-provided ideas and proof sketches. 
No proof was adopted as generated and
every proof reached its final form through the authors.
The authors verified the correctness and originality of all content, including references,
and take full responsibility for the paper.

\bibliographystyle{alphaurl}
\bibliography{references}

\end{document}